\documentclass[11pt,a4paper]{article}
\usepackage{mathrsfs}
\usepackage{}
\usepackage{bbm}
\usepackage{amssymb}
\usepackage{indentfirst,mathrsfs}
\usepackage{amsfonts,amsmath,amssymb,amsthm}
\usepackage{latexsym,amscd,multirow}
\usepackage{amsbsy}
\usepackage{xcolor}
\usepackage[all,cmtip]{xy}
\usepackage{graphicx}
\usepackage{longtable}
\usepackage{multirow}
\usepackage{array}
\usepackage{bm}
\usepackage{makecell}
\usepackage{booktabs}

\newtheorem{thm}{Theorem}
\newtheorem{lem}[thm]{Lemma}   
\newtheorem{cor}[thm]{Corollary}

\newtheorem{example}{Example}
\newtheorem{remark}{Remark}

\newcounter{alg}
\newlength{\lefttab}
\newlength{\numberoffset}
{\trivlist
   \topsep=0pt\parsep=0pt\itemsep=0pt
   \addtolength{\lefttab}{1.25em}
   \addtolength{\numberoffset}{1.25em}
   \leftskip=\lefttab}%
  {\endtrivlist}

\begin{document}
\title{Constructions of $q$-ary Golay Complementary Pairs Over Flexible Non-Power-of-Two Lengths.}

\author{Zhiye Yang$^1$ and Keqin Feng$^2$\\
~\\
$^1$ Research Center for Number Theory and Its Applications\\
School of Mathematics, Northwest University\\ Xi'an 710127,  Shaanxi, China\\
E-mail address: zyyang02@126.com \\
$^2$ Department of Mathematical Sciences,Tsinghua University,\\ Beijing, 100084, China.\\
E-mail address: fengkq@tsinghua.edu.cn\\}

\date{}

\maketitle
\begin{quote}
{\small {\bf Abstract:}}
Golay complementary pair (GCP), first introduced by Golay in 1951, has been extensively studied and widely applied in communication systems. 
 A $q$-ary GCP $\{\mathbf{A},\mathbf{B}\}$ consists of two $q$-ary complex sequences $\mathbf{A}=(A_0,\cdots,A_{M-1})$ and $\mathbf{B}=({B}_0,\cdots,{B}_{M-1})$ of equal length $M$, where $\textit{A}_i,\textit{B}_i\in\{\xi^a:0\leq a\leq q-1\}$ with $\xi=e^{\frac{2\pi\sqrt{-1}}{q}}$. 
 In this paper, 
  we prove that the existence of a quaternary ($q=4$) GCP of length $M$ is equivalent to the explicit constructibility of ($4h$)-ary GCPs of length $2^mM$ for all integers $h,m\geq1$. All proposed sequences are constructed via extended Boolean functions (EBFs), and the direct construction yields GCPs with more flexible length ranges than all previous relevant results.




\end{quote}

\section{Introduction}
Sequences with favorable correlation properties are  essential for high-resolution target sensing in integrated sensing and communication (ISAC), radar, and communication systems. Let $q\geq2$ be an integer. A $q$-ary sequence with length $M\geq2$ is a \emph{complex sequence} $\mathbf{A}=({A}_0,\cdots,{A}_{M-1})$ with ${A}_i=\xi^{a_i}$ for all $0\leq i\leq M-1$, where $\xi=e^{\frac{2\pi\sqrt{-1}}{q}}$ and $a_i\in\mathbb{Z}_q=\{0,1,\cdots,q-1\}$. The \emph{aperiodic auto-correlation function} of $\mathbf{A}$ at shift $\lambda$ is defined by
\begin{equation}\label{AACF}
C_\mathbf{A}(\lambda)=\left\{\begin{array}{cl}
\sum\limits_{l=0}^{M-1-\lambda}A_lA_{l+\lambda}^*, &0\leq\lambda\leq M-1, \\
\sum\limits_{l=0}^{M-1-\lambda}A_{l-\lambda}A_{l}^*, &-(M-1)\leq\lambda<0 , \\
0, &|\lambda|\geq M,
\end{array}\right.
\end{equation}
where $A_{l}^*$ is the complex conjugation of $A_{l}$. By definition, we have $C_\mathbf{A}(0)=\sum\limits_{l=0}^{M-1}A_lA_{l}^*=\sum\limits_{l=0}^{M-1}1=M$. For two $q$-ary sequences
$\mathbf{A}=(A_0,\cdots,A_{M-1})$ and $\mathbf{B}=({B}_0,\cdots,{B}_{M-1})$ of equal length $M$, the pair $\{\mathbf{A},\mathbf{B}\}$ is called a \emph{Golay complementary pair}, denoted by an $(M,q)$-GCP, if
\begin{equation}\label{GCPdef}
C_\mathbf{A}(\lambda)+C_\mathbf{B}(\lambda)=0, \,\,\,\text{for}\,1\leq\lambda\leq M-1.
\end{equation}
From the definition of (\ref{AACF}), it follows that $C_\mathbf{A}(\lambda)=C_\mathbf{A}(-\lambda)^*$. Consequently, if (\ref{GCPdef}) holds for all $1\leq\lambda\leq M-1$, then it holds for all $\lambda\neq0$.

The concept of Golay complementary pairs was first introduced by Marcel J. E. Golay in 1951\cite{Golay1951}, who further investigated their properties for the binary case ($q=2$) in 1961\cite{Golay1961tit}. Benefiting from their ideal correlation properties, GCPs have been extensively employed in a wide range of communication and signal processing systems, including channel estimation\cite{Spasojevic2001tit}, orthogonal frequency division multiplexing (OFDM)\cite{Paterson2000tit,Sarkar2020TCOMM}, radar waveform design\cite{Pezeshki2008tit}, pulse compression\cite{Budisin1991}, etc.

Many GCPs with various parameters $q$ and $M$ have been discovered via two primary approaches: computer search\cite{Borwein2004,Craigen2002DM}, and systematic construction using methods based on Hadamard matrices\cite{Seberry1992}, generalized Boolean functions\cite{Davis1999tit,Kumar2023TCOMM,Priyanshu2025CCDS}, and para-unitary matrices\cite{Wang2021tit}.  

For the binary case ($q=2$), $(M,2)$-GCPs (BGCPs) have been constructed \cite{Turyn1974} for following lengths 
$$M=2^a10^b26^c\geq2,\,\,\,a,b,c\geq0,$$
and it is conjectured that no BGCPs exist for any other lengths. 
For the quaternary case ($q=4$), 
 we have the following existence result: 
\begin{lem}\label{QGCPsExist}
	\cite[Corollary~6]{Craigen2002DM}
	There exist $(M,4)$-GCPs for lengths of the form
	$$M=2^{a+u}3^b5^c11^d13^e\geq2,$$
	where $a,b,c,d,e,u\geq0,\,b+c+d+e\leq a+2u+1$, and $u\leq c+e.$
\end{lem}
For the construction of quaternary Golay complementary pairs (QGCPs), Davis and Jedwab in \cite{Davis1999tit} provided an explicit construction of $(2^h,2^m)$-GCPs for all $h,m\geq1$ using generalized Boolean functions (GBFs). Recently, by employing restricted Boolean functions, Kumar et al. \cite{Kumar2023TCOMM} proposed a direct construction of $(M,q)$-GCPs over non-power-of-two lengths $M=5\cdot2^{m-3}\,(m\geq5)$ and 
$M=13\cdot2^{m-4}\,(m\geq6)$ for all even $q$ , which settles the open problem of direct GCP construction over non-power-of-two lengths. In 2025, Priyanshu et al. \cite{Priyanshu2025CCDS} further presented a direct construction of $(3\cdot2^m,4h)$-GCPs and $(11\cdot2^m,4h)$-GCPs  for all $h,m\geq1$, based on extended Boolean functions (EBFs).


 However, a critical limitation of these existing constructions is that they are typically restricted to specific lengths, which motivates this paper to pursue more flexible and general GCP constructions. Motivated by the works in \cite{Kumar2023TCOMM} and \cite{Priyanshu2025CCDS}, in this paper we demonstrate that the construction framework in \cite{Priyanshu2025CCDS} can be generalized to support substantially more flexible lengths. 
 Specifically, we prove that a $(M,4)$-GCP exists if and only if $(M\cdot2^m,4h)$-GCPs exist for all integers $h,m\geq1$. 
The result shows that $M$ can be any integer given in Lemma \ref{QGCPsExist}. Investigating this problem is of practical significance for designing sequences with low peak-to-mean envelope power ratio (PMEPR) under flexible length requirements.

The remainder of this paper is organized as follows. Section 2 gives some preliminaries. In Section 3, we present the main results of this paper and give some examples for illustration. Section 4 concludes this paper.

\section{Preliminaries}
In this section, we recall the definition of generalized Boolean function (GBF), extended Boolean function (EBF) and the fundamental relationship between GBFs, EBFs and $q$-ary complex sequences.


Let $\mathbb{Z}_2=\{0,1\}$, $\mathbb{Z}_{4h}=\{0,1,\cdots,4h-1\}$ for any integer $h\geq1$, and let $m\geq2$ be an integer. A \emph{generalized Boolean function} (GBF) 
$$f=f(x_1,x_2,\cdots,x_m):\mathbb{Z}_2^m\rightarrow\mathbb{Z}_{4h}$$
can be represented by its \emph{value sequence} of length $2^m$ over $\mathbb{Z}_{4h}$:
$$\boldsymbol{f}=\left(f_0,f_1,\cdots,f_{2^m-1} \right),\,\,\,f_I\in \mathbb{Z}_{4h},\,\,\,0\leq I\leq2^m-1,$$
where $f_I\in \mathbb{Z}_{4h}$ for all $0\leq I\leq2^m-1$. For each index $I$, its $2$-adic expansion is given by
$$I=i_12^{m-1}+i_22^{m-2}+\cdots+i_{m-1}2+i_m,\,\,\,i_k\in\left\lbrace 0,1\right\rbrace,\,\,1\leq k\leq m,$$
and then we define $f_I=f(i_1,i_2,\cdots,i_{m})$ and associate $\boldsymbol{f}$ with a \emph{complex-valued sequence} 
$$\mathbf{F}=\left(\xi^{f_0},\xi^{f_1},\cdots,\xi^{f_{2^m-1}} \right)$$  
where $\xi=e^\frac{2\pi\sqrt{-1}}{4h}$.

For any integer $M\geq2$, an \emph{extended Boolean function} (EBF) is a generalization of a GBF, defined as follows 
$$f=f(x_1,x_1,\cdots,x_m,y):\mathbb{Z}_2^m\times\mathbb{Z}_M\rightarrow\mathbb{Z}_{4h}.$$
This function can be represented by its value sequence of length $M\cdot2^m$ over $\mathbb{Z}_{4h}$
$$\boldsymbol{f}=\left(f_0,f_1,\cdots,f_{M\cdot2^m-1} \right),\,\,\,f_I\in \mathbb{Z}_{4h},\,\,\,0\leq I\leq M\cdot2^m-1,$$
where
\begin{align*}
	\left\{
	\begin{array}{ll}
		I=I'M+y, & 0\leq I'\leq 2^m-1,\,0\leq y\leq M-1,\\
		I'=i_12^{m-1}+i_22^{m-2}+\cdots+i_{m-1}2+i_m, & i_k\in\left\lbrace 0,1\right\rbrace,\,1\leq k\leq m,
	\end{array}
	\right.
\end{align*}
We then set $f_I=f(i_1,i_2,\cdots,i_{m},y)$ and, similarly, associate this value sequence $\boldsymbol{f}$ with a complex-valued sequence
$$\mathbf{F}=\left(\xi^{f_0},\xi^{f_1},\cdots,\xi^{f_{M\cdot2^m-1}} \right)$$  
where $\xi=e^\frac{2\pi\sqrt{-1}}{4h}$.

In the following, we always identify $I$ with $(i_1,i_2,\cdots,i_{m},y)$. We also denote $\xi=e^\frac{2\pi\sqrt{-1}}{4h}$ and $\omega=e^\frac{2\pi\sqrt{-1}}{4}$, where $\xi$ and $\omega$ are primitive $4h$-th and $4$-th roots of unity, respectively.

%

\section{Proposed Construction of $q$-ary GCPs}
\subsection{Main Theorem: Necessary and Sufficient Condition for GCP Construction}
Let $h\geq1$, $m,M\geq2$, and let $\pi$ be a permutation on $\{1,2,\cdots,m\}$. For two 
generalized Boolean functions 
$$\phi_1,\phi_2:\mathbb{Z}_M\rightarrow\mathbb{Z}_{4}.$$
we define three extended Boolean functions 
$$f,a,b:\mathbb{Z}_2^m\times\mathbb{Z}_M\rightarrow\mathbb{Z}_{4h}$$
as follows
\begin{equation}\label{EBFdef}
	\left\{\begin{array}{ll}
		f(x_1,x_2,\cdots,x_{m},y)=2h\sum\limits_{k=1}^{m-1}x_{\pi(k)}x_{\pi(k+1)}+\sum\limits_{k=1}^{m}c_kx_k+h\left[ x_{\pi(m)}(\phi_2(y)-\phi_1(y))+\phi_1(y)\right] ,\\
		a(x_1,x_2,\cdots,x_{m},y)=f(x_1,x_2,\cdots,x_{m},y)+\theta,\\
		b(x_1,x_2,\cdots,x_{m},y)=f(x_1,x_2,\cdots,x_{m},y)+2hx_{\pi(1)}+\theta'.
	\end{array}\right.
\end{equation}
where $\theta,\theta',c_k\in\mathbb{Z}_{4h}$ for $1\leq k\leq m$.

Let 
\begin{equation}\label{seqdefphi}
\mathbf{\Phi}_1=(\Phi_1(0),\Phi_1(1),\cdots,\Phi_1(M-1)) \quad \text{and} 
\quad \mathbf{\Phi}_2=(\Phi_2(0),\Phi_2(1),\cdots,\Phi_2(M-1))
\end{equation}
be the quaternary sequences associated with $\phi_1$ and $\phi_2$, respectively, where $\Phi_1(i)=\omega^{\phi_1(i)}$,  $\Phi_2(i)=\omega^{\phi_2(i)}$ and $\omega=e^\frac{2\pi\sqrt{-1}}{4}$ for $0\leq i\leq M-1$.
Let 
$$\boldsymbol{a}=\left(a_0,a_1,\cdots,a_{M\cdot2^m-1} \right)\,\,\,\text{and}\,\,\,\boldsymbol{b}=\left(b_0,b_1,\cdots,b_{M\cdot2^m-1} \right)$$
be the sequences over $\mathbb{Z}_{4h}$ of length $M\cdot2^m$ corresponding to
the EBFs $a$ and $b$, respectively. We then define the $4h$-ary complex-valued sequences 
\begin{equation}\label{seqdefab}
\mathbf{A}=\left(\xi^{a_0},\xi^{a_1},\cdots,\xi^{a_{M\cdot2^m-1}} \right)\,\,\,\text{and}\,\,\,\mathbf{B}=\left(\xi^{b_0},\xi^{b_1},\cdots,\xi^{b_{M\cdot2^m-1}} \right),\,\,\,\xi=e^\frac{2\pi\sqrt{-1}}{4h},
\end{equation}
 and both $\mathbf{A}$ and $\mathbf{B}$ have length $M\cdot2^m$.

From the above construction, we obtain the following main result.

\begin{thm}\label{if-th}
	Let $h\geq1$, $m,M\geq2$. Let $\mathbf{\Phi}_1, \mathbf{\Phi}_2$ and $\mathbf{A}, \mathbf{B}$ be the sequences derived via the construction in (\ref{seqdefphi}) and (\ref{seqdefab}), respectively. 
	Then $\{\mathbf{A},\mathbf{B}\}$ is a $(M\cdot2^m,4h)$-GCP if and only if $\{\mathbf{\Phi}_1,\mathbf{\Phi}_2\}$ is a $(M,4)$-GCP.
\end{thm}
\begin{proof}
Let $L=M\cdot2^m$. We need to prove that for all $1\leq\lambda\leq L-1$,
\begin{equation}\label{CAF-GCP}
\sum\limits_{I=0}^{L-1-\lambda}\xi^{{a}_I-{a}_{I+\lambda}}+\sum\limits_{I=0}^{L-1-\lambda}\xi^{{b}_I-{b}_{I+\lambda}}=0
\end{equation}
if and only if $\{\mathbf{\Phi}_1,\mathbf{\Phi}_2\}$ is a Golay complementary pair (GCP).

For each integer $I$ with $0\leq I\leq L-1-\lambda$, let 
$$I=(i_1,i_2,\cdots,i_{m},\beta)\,\,\,\text{and}\,\,\,I+\lambda=(j_1,j_2,\cdots,j_{m},\beta')$$
where $i_k,j_k\in\left\lbrace 0,1\right\rbrace$ for $1\leq k \leq m$, and 
$0\leq \beta,\beta'\leq M-1$.

The integer interval $D=\left[0,L-1-\lambda \right]$ can be partitioned into the following three disjoint subsets
\[
\begin{aligned}
	D_1&=\bigl\{I\in D : i_{\pi(1)}=1-j_{\pi(1)}\bigr\}, \text{ namely, } \bigl(i_{\pi(1)},j_{\pi(1)}\bigr)=(0,1) \text{ or } (1,0), \\[4pt]
	D_2&=\bigl\{I\in D : i_{\pi(1)}=j_{\pi(1)},\ \exists u\in\{2,\dots,m\}, \\
	&\quad\ i_{\pi(k)}=j_{\pi(k)},\ \forall 1\leq k\leq u-1,\ i_{\pi(u)}=1-j_{\pi(u)}\bigr\}, \\[4pt]
	D_3&=\bigl\{I\in D : (i_1,i_2,\dots,i_{m})=(j_1,j_2,\dots,j_{m})\bigr\}, \text{ so that } \\
	&\quad\ \bigl(i_{\pi(1)},i_{\pi(2)},\dots,i_{\pi(m)}\bigr)=\bigl(j_{\pi(1)},j_{\pi(2)},\dots,j_{\pi(m)}\bigr).
\end{aligned}
\]
We now show that 
\begin{equation*}
	\sum\limits_{I\in D_r}(\xi^{{a}_I-{a}_{I+\lambda}}+\xi^{{b}_I-{b}_{I+\lambda}})=0 \quad\text{for}\,\,r=1,2
\end{equation*}
and that
\begin{equation*}
	\sum\limits_{I\in D_3}(\xi^{{a}_I-{a}_{I+\lambda}}+\xi^{{b}_I-{b}_{I+\lambda}})=0 
\end{equation*}
if and only if $\{\mathbf{\Phi}_1,\mathbf{\Phi}_2\}$ is a GCP, which together imply the desired equality in (\ref{CAF-GCP}).

Case 1:\,\,$I\in D_1$\\
By definition, $I=(i_1,i_2,\cdots,i_{m},\beta)$, $I+\lambda=(j_1,j_2,\cdots,j_{m},\beta')$ and $\bigl(i_{\pi(1)},j_{\pi(1)}\bigr)=(0,1)$ or $(1,0)$. From the definition of the functions $f$, $a$ and $b$ in (\ref{EBFdef}),
we get 
\begin{eqnarray*}\nonumber
	{b}_I-{a}_I
	&= &f(i_1,i_2,\cdots,i_{m},\beta)+2hi_{\pi(1)}+\theta'-(f(i_1,i_2,\cdots,i_{m},\beta)+\theta)\nonumber\\
	&= &2hi_{\pi(1)}+\theta'-\theta.
\end{eqnarray*}
Similarly, we have ${b}_{I+\lambda}-{a}_{I+\lambda}=2hj_{\pi(1)}+\theta'-\theta$. Combining these results, we obtain
$$({b}_I-{b}_{I+\lambda})-({a}_I-{a}_{I+\lambda})=({b}_I-{a}_{I})-({b}_{I+\lambda}-{a}_{I+\lambda})=2h(i_{\pi(1)}-j_{\pi(1)})\equiv 2h\,\,(\bmod\,4h).$$
Since $\xi^{2h}=-1$,  it follows that
$$\xi^{{a}_I-{a}_{I+\lambda}}+\xi^{{b}_I-{b}_{I+\lambda}}=\xi^{{a}_I-{a}_{I+\lambda}}+\xi^{{a}_I-{a}_{I+\lambda}+2h}=0.$$
Summing over all $I\in D_1$, we conclude that $\sum\limits_{I\in D_1}(\xi^{{a}_I-{a}_{I+\lambda}}+\xi^{{b}_I-{b}_{I+\lambda}})=0.$

Case 2:\,\,$I\in D_2$\\
By definition, $I=(i_1,i_2,\cdots,i_{m},\beta)$, $I+\lambda=(j_1,j_2,\cdots,j_{m},\beta')$ where there exists some integer $u$, $2\leq u\leq m$ such that $i_{\pi(k)}=j_{\pi(k)}$ for all $1\leq k\leq u-1$, and $ i_{\pi(u)}\neq j_{\pi(u)}$.
Define 
$$I'=(i'_1,i'_2,\cdots,i'_{m},\beta)\,\,\,\text{and}\,\,\,J'=(j'_1,j'_2,\cdots,j'_{m},\beta')$$
where 
\begin{equation}\label{IJdef}
	\begin{aligned}
		i'_{\pi(k)} &= i_{\pi(k)}, & j'_{\pi(k)} &= j_{\pi(k)}, && \text{for } 1\leq k \leq m,\,k\neq u-1, \\
		i'_{\pi(u-1)} &= 1-i_{\pi(u-1)}, & j'_{\pi(u-1)} &= 1-j_{\pi(u-1)}.
	\end{aligned}
\end{equation}
Since $i_{\pi(u-1)}=j_{\pi(u-1)}$, which implies $i'_{\pi(u-1)}=j'_{\pi(u-1)}$. Thus, we have 
\begin{eqnarray*}\nonumber
	I'-I
	&= &(0,\cdots,0,i'_{\pi(u-1)}-i_{\pi(u-1)},0,\cdots,0)\nonumber\\
	&= &(0,\cdots,0,j'_{\pi(u-1)}-j_{\pi(u-1)},0,\cdots,0)\nonumber\\
	&= &J'-(I+\lambda),
\end{eqnarray*}
and so $J'=I'+\lambda$. Namely, $I'\in D_2$ and $I'\neq I$.
From (\ref{EBFdef}) we get 
\begin{eqnarray*}\nonumber
	a_{I'}-a_{I}
	&= &a(i'_1,i'_2,\cdots,i'_{m},\beta)-a(i_1,i_2,\cdots,i_{m},\beta)\nonumber\\
	&= &2h(i_{\pi(u-2)}+i_{\pi(u)})(i'_{\pi(u-1)}-i_{\pi(u-1)})+c_{\pi(u-1)}(i'_{\pi(u-1)}-i_{\pi(u-1)})\nonumber\\
	&\equiv&2h(i_{\pi(u-2)}+i_{\pi(u)})+c_{\pi(u-1)}(1-2i_{\pi(u-1)})\,\,(\bmod\,4h). 
\end{eqnarray*}
Similarly, we have 
\begin{eqnarray*}\nonumber
	a_{I'+\lambda}-a_{I+\lambda}
	&= &a(j'_1,j'_2,\cdots,j'_{m},\beta)-a(j_1,j_2,\cdots,j_{m},\beta)\nonumber\\
	&\equiv&2h(j_{\pi(u-2)}+j_{\pi(u)})+c_{\pi(u-1)}(1-2j_{\pi(u-1)})\,\,(\bmod\,4h). 
\end{eqnarray*}
If $u=2$, let $j_{\pi(0)}=i_{\pi(0)}=0$. Thus 
\begin{eqnarray*}\nonumber
	(a_{I'}-a_{I'+\lambda})-(a_{I}-a_{I+\lambda})
	&= &(a_{I'}-a_{I})-(a_{I'+\lambda}-a_{I+\lambda})\nonumber\\
	&\equiv&2h(i_{\pi(u)}-j_{\pi(u)})\,\,(\bmod\,4h)\\
	&\equiv &2h\,\,(\bmod\,4h).
\end{eqnarray*}
Therefore 
$$\xi^{{a}_{I'}-{a}_{I'+\lambda}}+\xi^{{a}_I-{a}_{I+\lambda}}=\xi^{{a}_I-{a}_{I+\lambda}+2h}+\xi^{{a}_I-{a}_{I+\lambda}}=0.$$
Since set $D_2$ is the union of $\frac{|D_2|}{2}$ disjoint pairs $\{I,I'\}$,
it follows that $\sum\limits_{I\in D_2}\xi^{{a}_I-{a}_{I+\lambda}}=0.$
From the definition in (\ref{EBFdef}) we know that $b_{I'}-b_{I}=a_{I'}-a_{I}$ and $b_{I'+\lambda}-b_{I+\lambda}=a_{I'+\lambda}-a_{I+\lambda}$. Thus, by a similar argument, we also obtain 
$\sum\limits_{I\in D_2}\xi^{{b}_I-{b}_{I+\lambda}}=0$. 
And then, $\sum\limits_{I\in D_2}(\xi^{{a}_I-{a}_{I+\lambda}}+\xi^{{b}_I-{b}_{I+\lambda}})=0.$

Case 3:\,\,$I\in D_3$\\
By definition, $I=(i_1,i_2,\cdots,i_{m},\beta)$ and  $I+\lambda=(i_1,i_2,\cdots,i_{m},\beta')$ where $(i_1,i_2,\cdots,i_{m})\in\mathbb{Z}_{2}^m$, $\beta'=\beta+\lambda$ and $0\leq \beta,\beta'\leq M-1$.
Therefore, $1\leq \lambda=\beta'-\beta\leq M-1$. From the definition in (\ref{EBFdef}) we obtain 
\[
\begin{aligned}
	b_{I}-b_{I+\lambda} &= a_{I}-a_{I+\lambda} = f(i_1,i_2,\dots,i_{m},\beta)-f(i_1,i_2,\dots,i_{m},\beta')\\
	&= h\bigl[i_{\pi(m)}\bigl(\phi_2(\beta)-\phi_1(\beta)\bigr)+\phi_1(\beta)-i_{\pi(m)}\bigl(\phi_2(\beta')-\phi_1(\beta')\bigr)-\phi_1(\beta')\bigr]\\
	&=
	\begin{cases}
		h(\phi_1(\beta)-\phi_1(\beta')), & \text{if}\,\,i_{\pi(m)}=0,\\[4pt]
		h(\phi_2(\beta)-\phi_2(\beta')), & \text{if}\,\,i_{\pi(m)}=1.
	\end{cases}
\end{aligned}
\]
Then for $1\leq \lambda\leq M-1$, we get
 \begin{eqnarray*}\nonumber
 	\sum\limits_{I\in D_3}\xi^{{b}_I-{b}_{I+\lambda}}
 	&= &\sum\limits_{I\in D_3}\xi^{{a}_I-{a}_{I+\lambda}}\nonumber\\
 		&=&\sum\limits_{i_{\pi(1)},\cdots,i_{\pi(m-1)}=0}^1\sum\limits_{\beta=0}^{M-1-\lambda}\xi^{h(\phi_1(\beta)-\phi_1(\beta'))}+\xi^{h(\phi_2(\beta)-\phi_2(\beta'))}\\
 	&=&2^{m-1}\sum\limits_{\beta=0}^{M-1-\lambda}\omega^{\phi_1(\beta)-\phi_1(\beta')}+\omega^{\phi_2(\beta)-\phi_2(\beta')}\\
 	&=&2^{m-1}\sum\limits_{\beta=0}^{M-1-\lambda}\omega^{\phi_1(\beta)-\phi_1(\beta+\lambda)}+\omega^{\phi_2(\beta)-\phi_2(\beta+\lambda)}\\
 	&\equiv &2^{m-1}\sum\limits_{\beta=0}^{M-1-\lambda}\mathbf{\Phi}_1(\beta)\mathbf{\Phi}_1^*(\beta+\lambda)+\mathbf{\Phi}_2(\beta)\mathbf{\Phi}_2^*(\beta+\lambda).
 \end{eqnarray*}
 which implies that $\sum\limits_{I\in D_3}(\xi^{{a}_I-{a}_{I+\lambda}}+\xi^{{b}_I-{b}_{I+\lambda}})=0$ if and only if $\{\mathbf{\Phi}_1,\mathbf{\Phi}_2\}$ forms a Golay complementary pair.
 
Combining the above results for $D_1$, $D_2$ and $D_3$, we obtain that $$\sum\limits_{I\in D_1}(\xi^{{a}_I-{a}_{I+\lambda}}+\xi^{{b}_I-{b}_{I+\lambda}})+\sum\limits_{I\in D_2}(\xi^{{a}_I-{a}_{I+\lambda}}+\xi^{{b}_I-{b}_{I+\lambda}})+\sum\limits_{I\in D_3}(\xi^{{a}_I-{a}_{I+\lambda}}+\xi^{{b}_I-{b}_{I+\lambda}})=0,$$
if and only if $\{\mathbf{\Phi}_1,\mathbf{\Phi}_2\}$ forms a GCP, which completes the proof of Theorem \ref{if-th}.

\end{proof}
\begin{remark}\label{Rem=1}
	For the case $m=1$, we set $\sum\limits_{k=1}^{m-1}x_{\pi(k)}x_{\pi(k+1)}=0$ and define
	\begin{equation*}
		\begin{cases}
			f(x_1,y)=c_1x_1+h\left[ x_1(\phi_2(y)-\phi_1(y))+\phi_1(y)\right],\\
			a(x_1,y)=f(x_1,y)+\theta,\\
			b(x_1,y)=f(x_1,y)+2hx_1+\theta'.
		\end{cases}
	\end{equation*}
	It can be similarly verified that the proof of Theorem \ref{if-th} still holds for $m=1$.
\end{remark}
Based on Lemma \ref{QGCPsExist}, Theorem \ref{if-th} yields the following immediate result.
\begin{cor}
	There exist $(M\cdot2^m,4h)$-GCPs for all integers $m,h\geq1$, where
	\[
	M=2^{a+u}3^b5^c11^d13^e\geq2,
	\]
	with $a,b,c,d,e,u\geq0$ satisfying
	\[
	b+c+d+e\leq a+2u+1 \quad \text{and} \quad u\leq c+e.
	\]
\end{cor}
\subsection{Examples}
In this subsection, for illustration, we show two examples of GCPs obtained by the construction given in Theorem \ref{if-th}.
\begin{example}  
	Let $M=18$ and $h\geq 1$. We construct a $(2M,4h)$-GCP 
	$$\mathbf{A}=\left(A_{0},A_{1},\cdots,A_{35} \right)\,\,\,\text{and}\,\,\,\mathbf{B}=\left(B_{0},B_{1},\cdots,B_{35} \right)$$
	with length $2M=36$, which is divisible by \(9\), the square of an odd prime. For each $0\leq I\leq35$, we set  
	$$A_I=\xi^{a_I}\,\,\,\text{and}\,\,\,B_I=\xi^{b_I},\,\,\,\xi=e^\frac{2\pi\sqrt{-1}}{4h}$$ 
	where 
	$$\boldsymbol{a}=\left(a_0,a_1,\cdots,a_{35} \right)\,\,\,\text{and}\,\,\,\boldsymbol{b}=\left(b_0,b_1,\cdots,b_{35} \right)$$
	are sequences over $\mathbb{Z}_{4h}$. 
	To employ the construction in Theorem \ref{if-th}, we need to find two functions
	\[
	\phi_1,\phi_2:\mathbb{Z}_{18}\rightarrow\mathbb{Z}_{4}
	\]
	such that the sequences
	\[
	\boldsymbol{\Phi}_1=\bigl(\Phi_1(0),\Phi_1(1),\dots,\Phi_1(17)\bigr),\quad
	\boldsymbol{\Phi}_2=\bigl(\Phi_2(0),\Phi_2(1),\dots,\Phi_2(17)\bigr)
	\]
	with \(\Phi_1(i)=\omega^{\phi_1(i)}\), \(\Phi_2(i)=\omega^{\phi_2(i)}\), and \(\omega=e^{\frac{2\pi\sqrt{-1}}{4}}\) for \(0\leq i\leq17\), form a quaternary Golay complementary pair (QGCP) of length \(18\). The existence of such a QGCP is guaranteed by Lemma \ref{QGCPsExist} with $a=1,b=2,c=d=u=e=0$. Several constructions are available in the literature; see, e.g., Proposition 9 in \cite{Gibson2011} or Theorem 5 in \cite{Fiedler2008JCT}.
	One can verify that for the following functions $\phi_1,\phi_2:\mathbb{Z}_{18}\rightarrow\mathbb{Z}_{4}$ defined in Table \ref{tab:phi} satisfy
	\begin{equation*}
		C_{\mathbf{\Phi}_1}(\lambda)+C_{\mathbf{\Phi}_2}(\lambda)	=\sum\limits_{I=0}^{17-\lambda}\omega^{\phi_1(I)-\phi_1(I+\lambda)}+\omega^{\phi_2(I)-\phi_2(I+\lambda)}=0,\,\,\,1\leq\lambda\leq17.
	\end{equation*} 
	This implies that $\{\boldsymbol{\Phi}_1,\boldsymbol{\Phi}_2\}$ is indeed a $(18,4)$-GCP.
	\begin{table}[htp]
		\centering
		\caption{Functions $\phi_1$ and $\phi_2$ generating an $(18,4)$-QGCP.}
		\renewcommand{\arraystretch}{1.2} 
		\begin{tabular}{@{}c*{18}{c}@{}} 
			\toprule 
			$y$   & 0 & 1 & 2 & 3 & 4 & 5 & 6 & 7 & 8 & 9 & 10 & 11 & 12 & 13 & 14 & 15 & 16 & 17 \\
			\midrule  
			$\phi_1(y)$ & 0 & 0 & 2& 0 & 0 & 2 & 2 & 2 & 0 & 0 & 3 & 0 & 1 & 0 & 1 & 0 & 3 & 0\\
			$\phi_2(y)$  & 0 & 1 & 0& 0 & 1 & 0 & 2 & 3 & 2 & 0 & 2 & 2 & 1 & 3 & 3 & 0 & 2 & 2  \\
			\bottomrule 
		\end{tabular}
		\label{tab:phi}
	\end{table}
	
	Now we employ the construction in Theorem \ref{if-th} with $m=1$ (see Remark \ref{Rem=1}).
	Define the functions $f,a,b:\mathbb{Z}_{2}\times\mathbb{Z}_{18}\rightarrow\mathbb{Z}_{4h}$ by 
		\begin{equation}\label{m=1EBFdef}
		\begin{cases}
			f(x,y)=x+h\left[ x(\phi_2(y)-\phi_1(y))+\phi_1(y)\right],\\
			a(x,y)=f(x,y)+1,\\
			b(x,y)=f(x,y)+2hx.
		\end{cases}
	\end{equation}
	From \eqref{m=1EBFdef}, for $0 \leq I = y \leq 17$, we obtain
	$$a_I=a(0,y)=f(0,y)+1=h\phi_1(y)+1,\quad b_I=f(0,y)=h\phi_1(y),$$
	For $18 \leq I = 18 + y \leq 35$, we have
	$$a_I=a(1,y)=f(1,y)+1=h\phi_2(y)+1,\quad b_I=f(1,y)+2h=h(\phi_2(y)+2)+1.$$
	This yields the sequences
	$$\boldsymbol{a}=\left(a_0,a_1,\cdots,a_{35} \right)\,\,\,\text{and}\,\,\,\boldsymbol{b}=\left(b_0,b_1,\cdots,b_{35} \right),\,\,\,a_I,b_I\in\mathbb{Z}_{4h}.$$
	 We then define the corresponding $4h$-ary complex-valued sequences
	$$\mathbf{A}=\left(A_{0},A_{1},\cdots,A_{35} \right)\,\,\,\text{and}\,\,\,\mathbf{B}=\left(B_{0},B_{1},\cdots,B_{35} \right)$$
    with $A_I=\xi^{a_I}$ and $B_I=\xi^{b_I}$, where $\xi=e^{\frac{2\pi\sqrt{-1}}{4h}}$. These sequences are presented in the following Table \ref{0-17seq} and \ref{18-35seq} below. 
 \begin{table}[htp]
 	\centering
 	\caption{Sequences $\boldsymbol{a},\boldsymbol{b}, \mathbf{A}, \mathbf{B}$ for $0\leq I\leq17$}
 	\label{0-17seq}
 	\renewcommand{\arraystretch}{1.4}
 	\footnotesize
 	\setlength{\tabcolsep}{4.6pt} 
 	\begin{tabular}{@{}c*{18}{c}@{}}
 		\toprule
 		$I$        & 0 & 1 & 2 & 3 & 4 & 5 & 6 & 7 & 8 & 9 & 10 & 11 & 12 & 13 & 14 & 15 & 16 & 17 \\
 		\midrule
 		$a_I$      & 1 & 1 & $2h\!+\!1$ & 1 & 1 & $2h\!+\!1$ & $2h\!+\!1$ & $2h\!+\!1$ & 1 & 1 & $3h\!+\!1$ & 1 & $h\!+\!1$ & 1 & $h\!+\!1$ & 1 & $3h\!+\!1$ & 1 \\
 		$b_I$      & 0 & 0 & $2h$ & 0 & 0 & $2h$ & $2h$ & $2h$ & 0 & 0 & $3h$ & 0 & $h$ & 0 & $h$ & 0 & $3h$ & 0 \\
 		$A_I$      & $\xi$ & $\xi$ & $-\xi$ & $\xi$ & $\xi$ & $-\xi$ & $-\xi$ & $-\xi$ & $\xi$ & $\xi$ & $-\omega\xi$ & $\xi$ & $\omega\xi$ & $\xi$ & $\omega\xi$ & $\xi$ & $-\omega\xi$ & $\xi$ \\
 		$B_I$      & $1$ & $1$ & $-\xi$ & $1$ & $1$ & $-1$ & $-1$ & $-1$ & $1$ & $1$ & $-\omega$ & $1$ & $\omega$ & $1$ & $\omega$ & $1$ & $-\omega$ & $1$ \\
 		\bottomrule
 	\end{tabular}
 \end{table}
	\begin{table}[htp]
		\centering
		\caption{Sequences $\boldsymbol{a},\boldsymbol{b}, \mathbf{A}, \mathbf{B}$ for $18\leq I\leq35$}
		\label{18-35seq}
		\renewcommand{\arraystretch}{1.4}
		\footnotesize
		\setlength{\tabcolsep}{1.3pt}
		\begin{tabular}{@{}c*{18}{c}@{}}
			\toprule
			$I$    & 18 & 19 & 20 & 21 & 22 & 23 & 24 & 25 & 26 & 27 & 28 & 29 & 30 & 31 & 32 & 33 & 34 & 35 \\
			\midrule
			$a_I$  & 2 & $h\!+\!2$ & 2 & 2 & $h\!+\!2$ & 2 & $2h\!+\!2$ & $3h\!+\!2$ & $2h\!+\!2$ & 2 & $2h\!+\!2$ & $2h\!+\!2$ & $h\!+\!2$ & $3h\!+\!2$ & $3h\!+\!2$ & 2 & $2h\!+\!2$ & $2h\!+\!2$ \\
			$b_I$  & $2h\!+\!1$ & $3h\!+\!1$ & $2h\!+\!1$ & $2h\!+\!1$ & $3h\!+\!1$ & $2h\!+\!1$ & $1$ & $h\!+\!1$ & $1$ & $2h\!+\!1$ & $1$ & $1$ & $3h\!+\!1$ & $h\!+\!1$ & $h\!+\!1$ & $2h\!+\!1$ & $1$ & $1$ \\
			$A_I$  & $\xi^2$ & $\omega\xi^{2}$ & $\xi^2$ & $\xi^2$ & $\omega\xi^{2}$ & $\xi^2$ & $-\xi^{2}$ & $-\omega\xi^{2}$ & $-\xi^{2}$ & $\xi^2$ & $-\xi^{2}$ & $-\xi^{2}$ & $\omega\xi^{2}$ & $-\omega\xi^{2}$ & $-\omega\xi^{2}$ & $\xi^2$ & $-\xi^2$ & $-\xi^2$ \\
			$B_I$  & $-\xi$ & $-\omega\xi$ & $-\xi$ & $-\xi$ & $-\omega\xi$ & $-\xi$ & $\xi$ & $\omega\xi$ & $\xi$ & $-\xi$ & $\xi$ & $\xi$ & $-\omega\xi$ & $\omega\xi$ & $\omega\xi$ & $-\xi$ & $\xi$ & $\xi$ \\
			\bottomrule
		\end{tabular}
	\end{table}
	
	For any integer $h\geq1$, one can readily verify that 
	\begin{equation*}
		C_{\mathbf{A}}(\lambda)+C_{\mathbf{B}}(\lambda)	=\sum\limits_{I=0}^{35-\lambda}A_IA_{I+\lambda}^*+B_IB_{I+\lambda}^*=0,\,\,\,1\leq\lambda\leq35,
	\end{equation*}
confirming that $\{\mathbf{A},\mathbf{B}\}$ is indeed a $(36,4h)$-GCP.
This is consistent with the result presented in Theorem \ref{if-th}.
	
\end{example}
\begin{example} 
	Construction of $(44,4)$-GCP.
	We first adopt a (11,4)-GCP $\{\boldsymbol{\Phi}_1,\boldsymbol{\Phi}_2\}$ from reference [9], which is given by 
	\begin{equation*}
		\Phi_1:(1,\omega,-1,1,-1,\omega,-\omega,-1,\omega,\omega,1),\quad \Phi_2:(1,1,-\omega,-\omega,-\omega,1,1,\omega,-1,1,-1)
	\end{equation*}
	where $\omega=e^{\frac{2\pi\sqrt{-1}}{4}}=i=\sqrt{-1}.$
	
	The corresponding functions $\phi_1,\phi_2:\mathbb{Z}_{11}\rightarrow\mathbb{Z}_{4}$ are defined as shown in the following table 
	  \begin{table}[htp]
	  	\centering
	  	\renewcommand{\arraystretch}{1.2} 
	  	\begin{tabular}{@{}c*{11}{c}@{}} 
	  		\toprule 
	  		$y$   & 0 & 1 & 2 & 3 & 4 & 5 & 6 & 7 & 8 & 9 & 10  \\
	  		\midrule  
	  		$\phi_1(y)$ & 0 & 1 & 2& 0 & 2 & 1 & 3 & 2 & 1 & 1 & 0 \\
	  		$\phi_2(y)$  & 0 & 0 & 3& 3 & 3 & 0 & 0 & 1 & 2 & 0 & 2  \\
	  		\bottomrule 
	  	\end{tabular}
	  \end{table}
	  
	  Next, we consider the functions 
	  $f(x_1,x_2,y),a(x_1,x_2,y),b(x_1,x_2,y):\mathbb{Z}_{2}^2\times\mathbb{Z}_{11}\rightarrow\mathbb{Z}_{4}$ with $x_1,x_2\in\{0,1\}$ and $0\leq y\leq10$, defined as follows 
	  \begin{equation*}
	  	\begin{cases}
	  		f(x_1,x_2,y)=2x_1x_2+x_1+x_2+x_2\left[ (\phi_2(y)-\phi_1(y))+\phi_1(y)\right],\\
	  		a(x_1,x_2,y)=f(x_1,x_2,y)+1,\\
	  		b(x_1,x_2,y)=f(x_1,x_2,y)+2x_1+1.
	  	\end{cases}
	  \end{equation*}
	 From these functions, we obtain two sequences of length $44=2^2\times11$
	 $$\boldsymbol{a}=\left(a_0,a_1,\cdots,a_{43} \right)\,\,\,\text{and}\,\,\,\boldsymbol{b}=\left(b_0,b_1,\cdots,b_{43} \right)$$
	 where $a_I,b_I\in\mathbb{Z}_{4}$ for all $0\leq I\leq43$. Corresponding to these integer sequences, we define two complex sequences
	 $$\mathbf{A}=\left(A_{0},A_{1},\cdots,A_{43} \right)\,\,\,\text{and}\,\,\,\mathbf{B}=\left(B_{0},B_{1},\cdots,B_{43} \right)$$
	 with $A_I=\omega^{a_I}$ and $B_I=\omega^{b_I}$, where $\omega=e^{\frac{2\pi\sqrt{-1}}{4}}=i=\sqrt{-1}$. For each $0\leq I\leq43$, the correspondence between $I$ and $(x_1,x_2,y)$ is established by $I=11(x_2+2x_1)+y$, 
	 where $x_1,x_2\in\{0,1\}$ and $0\leq y\leq10$. Under this correspondence, $a_I=a(x_1,x_2,y)$ and $b_I=b(x_1,x_2,y)$ hold. Specifically, the explicit expressions for $a_I$ and $b_I$ in different intervals of $I$ are derived as follows
	 \begin{itemize}
	 	\item When $0\leq I=y\leq10$, $a_I=a(0,0,y)=f(0,0,y)+1=\phi_1(y)+1=b_I$.
	 	
	 	\item When $11\leq I=y+11\leq21$, $a_I=a(0,1,y)=f(0,1,y)+1=\phi_2(y)+2=b_I$.
	 	
	 	\item When $22\leq I=y+22\leq32$, $a_I=a(1,0,y)=f(1,0,y)+1=\phi_1(y)+2$ and $b_I=f(1,0,y)+3=\phi_1(y)$.
	 	
	 	\item When $33\leq I=y+33\leq43$, $a_I=a(1,1,y)=f(1,1,y)+1=\phi_2(y)+1$ and $b_I=f(1,1,y)+3=\phi_2(y)+3$.
	 	
	 \end{itemize}
	 These sequences are presented in the following Tables \ref{0-21seq} and \ref{22-43seq}  below.
	 \begin{table}[h!]
	 	\centering
	 	\caption{Sequences $\boldsymbol{a},\boldsymbol{b}, \mathbf{A}, \mathbf{B}$ for $0 \leq I \leq 21$}
	 	\label{0-21seq}
	 	\renewcommand{\arraystretch}{1.4}
	 	\footnotesize
	 	\setlength{\tabcolsep}{4.6pt} 
	 	\begin{tabular}{@{}c*{23}{c}@{}}
	 		\toprule
	 			$I$       & 0 & 1 & 2 & 3 & 4 & 5 & 6 & 7 & 8 & 9 & 10& 11& 12& 13& 14& 15& 16& 17& 18& 19& 20& 21 \\
	 			\midrule
	 			$a_I$     & 1 & 2 & 3 & 1 & 3 & 2 & 0 & 3 & 2 & 2 & 1 & 2 & 2 & 1 & 1 & 1 & 2 & 2 & 3 & 0 & 2 & 0 \\
	 			$b_I$     & 1 & 2 & 3 & 1 & 3 & 2 & 0 & 3 & 2 & 2 & 1 & 2 & 2 & 1 & 1 & 1 & 2 & 2 & 3 & 0 & 2 & 0 \\
	 			$A_I$     & $\omega$ & $-1$ & $-\omega$ & $\omega$ & $-\omega$ & $-1$ & $1$ & $-\omega$ & $-1$ & $-1$ & $\omega$ & $-1$ & $-1$ & $\omega$ & $\omega$ & $\omega$ & $-1$ & $-1$ & $-\omega$ & $1$ & $-1$ & $1$ \\
	 			$B_I$     & $\omega$ & $-1$ & $-\omega$ & $\omega$ & $-\omega$ & $-1$ & $1$ & $-\omega$ & $-1$ & $-1$ & $\omega$ & $-1$ & $-1$ & $\omega$ & $\omega$ & $\omega$ & $-1$ & $-1$ & $-\omega$ & $1$ & $-1$ & $1$ \\
	 			\bottomrule 
	 		\end{tabular}
	 	\end{table}
	 	\begin{table}[h!]
	 		\centering
	 		\caption{Sequences $\boldsymbol{a},\boldsymbol{b}, \mathbf{A}, \mathbf{B}$ for $22 \leq I \leq 43$}
	 		\label{22-43seq}
	 		\renewcommand{\arraystretch}{1.4}
	 		\footnotesize
	 		\setlength{\tabcolsep}{3.6pt} 
	 		\begin{tabular}{@{}c*{23}{c}@{}}
	 			\toprule
	 				$I$       & 22& 23& 24& 25& 26& 27& 28& 29& 30& 31& 32& 33& 34& 35& 36& 37& 38& 39& 40& 41& 42& 43 \\
	 				\midrule
	 				$a_I$     & 2 & 3 & 0 & 2 & 0 & 3 & 1 & 0 & 3 & 3 & 2 & 1 & 1 & 0 & 0 & 0 & 1 & 1 & 2 & 3 & 1 & 3 \\
	 				$b_I$     & 0 & 1 & 2 & 0 & 2 & 1 & 3 & 2 & 1 & 1 & 0 & 3 & 3 & 2 & 2 & 2 & 3 & 3 & 0 & 1 & 3 & 1 \\
	 				$A_I$     & $-1$ & $-\omega$ & $1$ & $-1$ & $1$ & $-\omega$ & $\omega$ & $1$ & $-\omega$ & $-\omega$ & $-1$ & $\omega$ & $\omega$ & $1$ & $1$ & $1$ & $\omega$ & $\omega$ & $-1$ & $-\omega$ & $\omega$ & $-\omega$ \\
	 				$B_I$     & $1$ & $\omega$ & $-1$ & $1$ & $-1$ & $\omega$ & $-\omega$ & $-1$ & $\omega$ & $\omega$ & $1$ & $-\omega$ & $-\omega$ & $-1$ & $-1$ & $-1$ & $-\omega$ & $-\omega$ & $1$ & $\omega$ & $-\omega$ & $\omega$ \\
	 				\bottomrule 
	 			\end{tabular}
	 		\end{table}
	 
	 	For any integer $1\leq\lambda\leq43$, it can be readily verified that
	 	\begin{equation*}
	 		C_{\mathbf{A}}(\lambda)+C_{\mathbf{B}}(\lambda)	=\sum\limits_{I=0}^{43-\lambda}A_IA_{I+\lambda}^*+B_IB_{I+\lambda}^*=0,
	 	\end{equation*}
	 	confirming that $\{\mathbf{A},\mathbf{B}\}$ is indeed a $(44,4)$-GCP.
	 	This is consistent with the result presented in Theorem \ref{if-th}.
\end{example}
\begin{remark}
	In the example $5$ in [13], the authors attempt to present a $(44,4)$-GCP using $f(x_1,x_2,y),a(x_1,x_2,y),b(x_1,x_2,y):\mathbb{Z}_{2}^2\times\mathbb{Z}_{11}\rightarrow\mathbb{Z}_{4}$ for $x_1,x_2\in\{0,1\}$ and $0\leq y\leq10$, defined by 
	\begin{equation*}
		\begin{cases}
			f(x_1,x_2,y)=2x_1x_2+x_1+x_2+x_1\left[ (\chi_2(y)-\chi_1(y))+\chi_1(y)\right],\\
			a(x_1,x_2,y)=f(x_1,x_2,y)+1,\\
			b(x,y)=f(x_1,x_2,y)+2x_1+1.
		\end{cases}
	\end{equation*}
	where
	$\chi_1,\chi_2:\mathbb{Z}_{11}\rightarrow\mathbb{Z}_{4}$ are given by\\
	\begin{equation*}
		\begin{cases}
			\chi_1(y)=y+(y-2)\lceil \frac{y}{3}\rceil+\lceil \frac{y}{5}\rceil+2\lceil \frac{y}{9}\rceil+2\lceil \frac{y}{10}\rceil,\\
			\chi_2(y)=3\lceil \frac{y}{2}\rceil+(y-3)\lceil \frac{y}{4}\rceil+\lceil \frac{y}{10}\rceil.
				\end{cases}
	\end{equation*}
	The values of $\chi_1$ and $\chi_2$ are shown in following table 
	\begin{table}[h!]
		\centering
		\renewcommand{\arraystretch}{1.2}
		\begin{tabular}{@{}c*{11}{c}@{}}
			\toprule
			$y$        & 0  & 1  & 2  & 3  & 4  & 5  & 6  & 7  & 8  & 9   & 10  \\
			\midrule
			$\chi_1(y)$& 0  & 1  & 3  & 1  & 1 & 0 & 0 & 0 & 0 & 0  & 2  \\
			$\chi_2(y)$& 0  & 2  & 3  & 3  & 0  & 2 & 0 & 1 & 3 & 2  & 1  \\
			\bottomrule
		\end{tabular}
	\end{table}
	
	For $0\leq I\leq43$, the index correspondence is $I=11(x_2+2x_1)+y$,  
	with $x_1,x_2\in\{0,1\}$ and $0\leq y\leq10$, such that 
	 $a_I=a(x_1,x_2,y)=f(x_1,x_2,y)+1$ and $b_I=b(x_1,x_2,y)=f(x_1,x_2,y)+2x_1+1$.
    Let $\omega=e^{\frac{2\pi\sqrt{-1}}{4}}=i=\sqrt{-1}$. 
    It is clamied in [13] that the complex sequences 
	$$\mathbf{A}=\left(A_{0},A_{1},\cdots,A_{43} \right)\,\,\,\text{and}\,\,\,\mathbf{B}=\left(B_{0},B_{1},\cdots,B_{43} \right)$$
	defined by $A_I=\omega^{a_I}$ and $B_I=\omega^{b_I}$, form
	a $(44,4)$-GCP. It appears that some of the computations or prints presented in [13] may not be correct. 
	Taking $x_1=x_2=1$ (i.e., $33\leq I=y+33\leq43$), we have   $b_I=f(1,1,y)+3=\chi_2(y)+3$, which yields $$\left(b_{33},b_{34},\cdots,b_{43} \right)=\left(\chi_2(0)+3,\chi_2(1)+3,\cdots,\chi_2(10)+3 \right)=\left(3,1,\cdots,1,0 \right)$$
	and consequently $\left(B_{33},B_{34},\cdots,B_{43} \right)=\left(-\omega,\omega,\cdots,\omega,1 \right)$. This does not match the sequence $\left(-\omega,-\omega,\cdots,-\omega,1 \right)$ presented in [13]. Furthermore, the sequence $\mathbf{A}$ is shown in [13] with a shorter length $40$.
\end{remark}

\section{Conclusion}

In this paper, we investigate that the existence of a quaternary GCP of length $M$ is equivalent to the explicit constructibility of ($4h$)-ary GCPs of length $2^mM$ for all integers $h,m\geq1$. Our result directly yields GCPs with length parameters far more flexible than those achievable by existing methods, overcoming the limitations of conventional sequence design. We believe this result opens new possibilities for applying GCPs in modern communication systems, particularly in scenarios requiring diverse sequence lengths and alphabet sizes.

\section*{Acknowledgement}

The authors would like to thank the anonymous referees for their helpful comments and suggestions.


\begin{thebibliography}{99}

\bibitem{Golay1951}
 Golay, M.~J.~E.: Static multislit spectrometry and its application to the panoramic display of infrared spectra. J. Opt. Soc. Am. 41(7), 468--472 (1951)

\bibitem{Golay1961tit}
Golay, M.: Complementary series. IRE Trans. Inf. Theory. 7(2), 82--87 (1961)

\bibitem{Spasojevic2001tit}
Spasojevic, P., Georghiades, C.~N.: Complementary sequences for ISI channel estimation. IEEE Trans. Inf. Theory. 47(3), 1145--1152 (2001)


\bibitem{Paterson2000tit}
Paterson, K.~G.: Generalized Reed-Muller codes and power control in OFDM modulation. IEEE Trans. Inf. Theory 46(1), 104--120 (2000)

\bibitem{Sarkar2020TCOMM}
Sarkar, P., Majhi, S., Liu, Z.: A direct and generalized construction of polyphase complementary sets with low PMEPR and high code-rate for OFDM system. IEEE Trans. Commun. 68(10), 6245--6262 (2020)

\bibitem{Pezeshki2008tit}
Pezeshki, A., Calderbank, A.~R., Moran, W., Howard, S.~D.: Doppler resilient Golay complementary waveforms. IEEE Trans. Inf. Theory. 54(9), 4254--4266 (2008)

\bibitem{Budisin1991}
Budisin, S.: Efficient pulse compressor for Golay complementary sequences. Electron. Lett. 31, 219--220 (1991)

\bibitem{Borwein2004}
Borwein, P., Ferguson, R.: A complete description of Golay pairs for lengths up to 100. Math. Comput. 73, 967--985 (2004)

\bibitem{Craigen2002DM}
Craigen, R., Holzmann, W., Kharaghani, H.: Complex Golay sequences: Structure and applications. Discrete Math. 252, 73--89 (2002)

\bibitem{Seberry1992}
Seberry, J., Yamada, M.: Hadamard matrices, sequences, and block designs. In: Dinitz, J.H., Stinson, D.R. (eds.) Contemporary Design Theory: A Collection of Surveys. 431--560 (1992)

\bibitem{Davis1999tit}
Davis, J.~A., Jedwab, J.: Peak-to-mean power control in OFDM, Golay complementary sequences, and Reed--Muller codes. IEEE Trans. Inf. Theory. 45, 2397--2417 (1999)

\bibitem{Kumar2023TCOMM}
Kumar, P., Majhi, S., Paul, S.: A direct construction of Golay complementary pairs and binary complete complementary codes of length non-power of two. IEEE Trans. Commun. 71(3), 1352--1363 (2023)

\bibitem{Priyanshu2025CCDS}
Priyanshu, P., Roy, A., Paul, S., Majhi, S.: Constructions of non-binary Golay complementary pairs of new lengths. Cryptogr. Commun. 17, 809--822 (2025)

\bibitem{Wang2021tit}
Wang, Z., Ma, D., Gong, G., Xue, E.: New construction of complementary sequence (or array) sets and complete complementary codes. IEEE Trans. Inf. Theory. 67(7), 4902--4928 (2021)

\bibitem{Turyn1974}
TurynR.~J.: Hadamard matrices, Baumert-Hall units, four-symbol sequences, pulse compression, and surface wave encodings. J. Comb. Theory, Ser. A. 16(3), 313--333 (1974)

\bibitem{Gibson2011}
Gibson, R.~G., Jedwab, J.: Quaternary Golay sequence pairs I: even length. Des. Codes Cryptogr. 59(1--3), 131--146 (2011)

\bibitem{Fiedler2008JCT}
Fiedler, F., Jedwab, J., Parker, M.~G.: A multi-dimensional approach to the construction and enumeration of Golay complementary sequences. J. Comb. Theory Ser. A. 115(5), 753--776 (2008)



























\end{thebibliography}
\end{document}